\documentclass[a4paper,UKenglish]{lipics-v2018}

\newcommand{\modulo}{\operatorname{mod}{}}
\newtheorem{observation}{Observation}

\nolinenumbers

\title{Improved Space-Time Tradeoffs for $k$SUM}

\author{Isaac Goldstein}{Bar-Ilan University, Ramat Gan, Israel}{goldshi@cs.biu.ac.il}{}{This research is supported by the Adams Foundation of the Israel Academy of Sciences and Humanities.}
\author{Moshe Lewenstein}{Bar-Ilan University,Ramat Gan, Israel}{moshe@cs.biu.ac.il}{}{This work was partially supported by an ISF grant \#1278/16, BSF grant 2010437 and a GIF grant 1147/2011.}
\author{Ely Porat}{Bar-Ilan University,Ramat Gan, Israel}{porately@cs.biu.ac.il}{}{This work was partially supported by an ISF grant \#1278/16.}

\authorrunning{I. Goldstein, M. Lewenstein and E. Porat}
\Copyright{Isaac Goldstein, Moshe Lewenstein and Ely Porat}

\begin{document}

\maketitle

\thispagestyle{empty}

\begin{abstract}
In the $k$SUM problem we are given an array of numbers $a_1,a_2,...,a_n$ and we are required to determine if there are $k$ different elements in this array such that their sum is 0. This problem is a parameterized version of the well-studied SUBSET-SUM problem, and a special case is the $3$SUM problem that is extensively used for proving conditional hardness. Several works investigated the interplay between time and space in the context of SUBSET-SUM. Recently, improved time-space tradeoffs were proven for $k$SUM using both randomized and deterministic algorithms.

In this paper we obtain an improvement over the best known results for the time-space tradeoff for $k$SUM. A major ingredient in achieving these results is a general self-reduction from $k$SUM to $m$SUM where $m<k$, and several useful observations that enable this reduction and its implications. The main results we prove in this paper include the following:
(i) The best known Las Vegas solution to $k$SUM running in approximately $O(n^{k-\delta\sqrt{2k}})$ time and using $O(n^{\delta})$ space, for $0 \leq \delta \leq 1$. (ii) The best known deterministic solution to $k$SUM running in approximately $O(n^{k-\delta\sqrt{k}})$ time and using $O(n^{\delta})$ space, for $0 \leq \delta \leq 1$. (iii) A space-time tradeoff for solving $k$SUM using $O(n^{\delta})$ space, for $\delta>1$. (iv) An algorithm for $6$SUM running in $O(n^4)$ time using just $O(n^{2/3})$ space. (v) A solution to $3$SUM on random input using $O(n^2)$ time and $O(n^{1/3})$ space, under the assumption of a random read-only access to random bits.

\end{abstract}

\section{Introduction}
In the $k$SUM problem we are given an array of numbers $a_1,a_2,...,a_n$ and we are required to determine if there are $k$ different elements in this array such that their sum equals 0. This is a parameterized version of SUBSET-SUM, one of the first well-studied NP-complete problems, which also can be thought of as a special case of the famous KNAPSACK problem~\cite{Karp72}. A special case of $k$SUM is the $3$SUM problem which is extensively used to prove conditional lower bounds for many problems, including: string problems~\cite{AWW14,ACLL14,GKLP16,KPP16}, dynamic problems~\cite{AW14,Patrascu10}, computational geometry problems~\cite{BHP99,GO95}, graph problems~\cite{AL13,AWY15,KPP16} etc.

The $k$SUM problem can be trivially solved in $O(n^k)$ time using $\tilde{O}(1)$ space (for constant $k$), or in $O(n^{\lceil k/2 \rceil})$ time using $O(n^{\lceil k/2 \rceil})$ space. It is known that there is no solution to $k$SUM with $n^{o(k)}$ running time, unless the Exponential Time Hypothesis is false~\cite{PW10}. However, a central goal is to find the best tradeoff between time and space for $k$SUM. Specifically, it is interesting to have a full understanding of questions like: What is the best running time we can achieve by allowing at most \emph{linear} space? How can the running time be improved by using $O(n^2)$, $O(n^3)$ or $O(n^{10})$ space? Can we get any improvement over $O(n^k)$ running time for almost constant space or use less space for $O(n^{\lceil k/2 \rceil})$ time solution? What is the best time-space tradeoff for interesting special cases like $3$SUM?
Questions of this type guided a line of research work and motivate our paper.

One of the first works on the time-space tradeoff of $k$SUM and SUBSET-SUM is by Shamir and Schroeppel~\cite{SS79}. They showed a simple reduction from SUBSET-SUM to $k$SUM. Moreover, they presented a deterministic solution to $4$SUM running in $O(n^2)$ time using $O(n)$ space. They used these solution and reduction to present an $O^{*}(2^{n/2})$ time and $O^{*}(2^{n/4})$ space algorithm for SUBSET-SUM. Furthermore, they demonstrate a space-time tradeoff curve for SUBSET-SUM by a generalized algorithm.
More recently, a line of research work improved the space-time tradeoff of Shamir and Schroeppel by using randomization. This includes works by  Howgrave-Graham and Joux~\cite{HGJ10}, Becker et al.~\cite{BCJ11} and Dinur et al.~\cite{DDKS12} on random instances of SUBSET-SUM, and a matching tradeoff curve for worst-case instances of SUBSET-SUM by Austrin et al.~\cite{AKKM13}.

Wang~\cite{Wang14} used randomized techniques to improve the space-time tradeoff curve for $k$SUM. Specifically, he presented a Las Vegas randomized algorithm for $3$SUM running in $\tilde{O}(n^2)$ time using just $\tilde{O}(\sqrt{n})$ space. Moreover, for general $k$ he demonstrated a Monte Carlo algorithm for $k$SUM that uses $O(n^{\delta})$ space using approximately $O(n^{k-\delta\sqrt{2k}})$ time, for $0 \leq \delta \leq 1$. Lincoln et al.~\cite{LWWW16} achieved $O(n^2)$ time and $\tilde{O}(\sqrt{n})$ space \emph{deterministic} solution for $3$SUM. For general $k$SUM ($k \geq 4$), they obtained a deterministic algorithm running in $O(n^{k-3+4/(k-3)})$ time using linear space and $O(n^{k-2+2/k})$ time using $O(\sqrt{n})$ space.

Very recently, Bansal et al.~\cite{BGNV17} presented a randomized solution to SUBSET-SUM running in $O^{*}(2^{0.86n})$ time and using just polynomial space, under the assumption of a random read-only access to exponentially many random bits. This is based on an algorithm that determines whether two given lists of length $n$ with integers bounded by a polynomial in $n$ share a common value. This problem is closely related to $2$SUM and they proved it can be solved using $O(\log{n})$ space in significantly less than $O(n^2)$ time if no value occurs too often in the same list (under the assumption of a random read-only access to random bits). They also used this algorithm to obtain an improved solution for $k$SUM on random input.

Finally, it is worth mentioning that recent works by Goldstein et al.~\cite{GKLP17,GLP17} consider the space-time tradeoff of data structures variants of $3$SUM and other related problems.

\subsection{Our Results}
In this paper we improve the best known bounds for solving $k$SUM in both (Las Vegas) randomized and deterministic settings. A central component in our results is a general \emph{self-reduction} from $k$SUM to $m$SUM for $m<k$:

\begin{theorem}\label{thm:self_reduction}
There is a self-reduction from one instance of $k$SUM with $n$ integers in each array to $O(n^{(k/m-1)(m-\delta)})$ instances of $m$SUM (reporting) with $n$ integers in each array and $O(n^{(k/m-1)(m-\delta)})$ instances of $\frac{k}{m}$SUM with $n^{\delta}$ integers in each array, for any integer $m<k$ and $0<\delta \leq m$.
\end{theorem}

Moreover, we present several crucial observations and techniques that play central role in this reduction and other results of this paper.

For general $k$SUM we obtain the following results:

Using our self-reduction scheme and the ideas by Lincoln et al.~\cite{LWWW16}, we obtain a deterministic solution to $k$SUM that significantly improves over the deterministic algorithm by Lincoln at al.~\cite{LWWW16} that runs in $O(n^{k-3+4/(k-3)})$ time using linear space and $O(n^{k-2+2/k})$ time using $O(\sqrt{n})$ space:

\begin{theorem}
For $k \geq 2$, $k$SUM can be solved by a deterministic algorithm that runs in $O(n^{k-\delta g(k)})$ time using $O(n^{\delta})$ space, for $0 \leq \delta \leq 1$ and $g(k)\ge \sqrt{k} - 2$.
\end{theorem}

By allowing randomization we have the following result:

\begin{theorem}\label{thm:ksum_las_vegas}
For $k \geq 2$, $k$SUM can be solved by a Las Vegas randomized algorithm that runs in $O(n^{k-\delta f(k)})$ time using $O(n^{\delta})$ space, for $0 \leq \delta \leq 1$ and $f(k)\ge \sqrt{2k}-2$.
\end{theorem}

Our Las Vegas algorithm has the same running time and space as Wang's~\cite{Wang14} Monte Carlo algorithm. The idea is to modify his algorithm using the observations and techniques from our self-reduction scheme.

We also consider solving $k$SUM using $O(n^{\delta})$ space for $\delta>1$. Using our self-reduction technique and the algorithm from Theorem~\ref{thm:ksum_las_vegas}, we prove the following:

\begin{theorem}
For $k \geq 2$, $k$SUM can be solved by a Las Vegas algorithm that runs in $O(n^{k-\sqrt{\delta} f(k)})$ time using $O(n^{\delta})$ space, for $\frac{k}{4} \geq \delta > 1$ and $f(k) \ge \sqrt{2k}-2$.
\end{theorem}

Our self-reduction technique can also be applied directly to obtain improvements on the space-time tradeoff for special cases of $k$SUM. Especially interesting is the case of $6$SUM which can be viewed as a combination of the "easy" 4SUM and the "hard" $3$SUM. We obtain randomized algorithms solving $6$SUM in $O(n^3)$ time using $O(n^2)$ space and in $O(n^4)$ time using just $O(n^{2/3})$ space (and not $O(n)$ as known by previous methods~\cite{Wang14}).

Finally, combining our techniques with the techniques by Bansal et al.~\cite{BGNV17} we obtain improved space-time tradeoffs for some special cases of $k$SUM on random input, under the assumption of a random read-only access to random bits. One notable result of this flavour is a solution to $3$SUM on random input that runs in $O(n^2)$ time and $O(n^{1/3})$ space, instead of the $O(n^{1/2})$ space solutions known so far~\cite{LWWW16,Wang14}.

\section{Preliminaries}\label{sec:preliminaries}
In the basic definition of $k$SUM the input contains just one array. However, in a variant of this problem, which is commonly used, we are given $k$ arrays of $n$ numbers and we are required to determine if there are $k$ elements, one from each array, such that their sum equals 0. It is easy to verify that this variant is equivalent to $k$SUM in terms of time and space complexity. We also note that the choice of 0 is not significant, as it can be easily shown that the problem is equivalent in terms of time and space complexity even if we put any other constant $t$, called the \emph{target number}, instead of 0. Throughout this paper we consider $k$SUM with $k$ arrays and a target value $t$. We also consider the \emph{reporting} version of $k$SUM in which we need to report \emph{all} subsets of $k$ elements that sum up to 0 or some other constant $t$.

All the randomized algorithms in this paper solve $k$SUM on input arrays that contain \emph{integer} numbers. The target number $t$ is also assumed to be an integer. This assumption was also used in previous papers considering the space-time tradeoff for $k$SUM (see~\cite{Wang14}). The deterministic solution we present is the only one that works even for $k$SUM on real numbers.

Let $\mathcal{H}$ be a family of hash functions from $[u]$ to $[m]$ ($[u]$ is some unbounded universe). $\mathcal{H}$ is called {\em linear} if for any $h\in\mathcal{H}$ and any $x_1,x_2 \in [u]$, we have $h(x_1) + h(x_2) \equiv h(x_1+x_2) \; (\modulo m)$.
$\mathcal{H}$ is called {\em almost-linear} if for any $h\in\mathcal{H}$ and any $x_1,x_2 \in [u]$, we have
either $h(x_1) + h(x_2) \equiv h(x_1+x_2) +c_h \; (\modulo m)$, or $h(x_1) + h(x_2) \equiv h(x_1+x_2) + c_h +1 \; (\modulo m)$, where $c_h$ is an integer that depends only on the choice of $h$.
Throughout this paper we will assume that $h$ is linear as almost linearity will just add a constant factor cost to the running time and a change in the offsets which can be easily handled.
For a function $h:[u] \rightarrow [m]$ and a set $S\subset [u]$ where $|S|=n$, we say that $i\in [m]$ is an overflowed value of $h$ if $|\{x\in S : h(x) = i\}| > 3n/m$.
$\mathcal{H}$ is called {\em almost-balanced} if for a random $h\in \mathcal{H}$ and any set $S\subset [u]$ where $|S|=n$, the expected number of elements from $S$ that are mapped to overflowed values is $O(m)$ (for more details see~\cite{BDP05,Dietzfelbinger96,KPP16,Wang14}).
There are concrete constructions of hash families that are almost-linear and almost-balanced~\cite{KPP16,Wang14}.
In the  Las Vegas algorithms in this paper, we assume, in order for the presentation to be clear, that an almost-balanced hash function can become balanced (which means that there are no overflowed values at all). The full details of how this can be done in our Las Vegas algorithms appear in Appendix~\ref{apx:balanced}.

\section{Self-Reduction From $k$SUM to $m$SUM}\label{sec:self_reduction}

We demonstrate a general efficient reduction from a single instance of $k$SUM to many instances of $m$SUM (reporting) and $\lceil \frac{k}{m} \rceil$SUM for $m<k$:

\setcounter{theorem}{0}

\begin{theorem}\label{thm:self_reduction}
There is a self-reduction from one instance of $k$SUM with $n$ integers in each array to $O(n^{(\lceil k/m \rceil -1)(m-\delta)})$ instances of $m$SUM (reporting) with $n$ integers in each array and $O(n^{(\lceil k/m \rceil -1)(m-\delta)})$ instances of $\lceil \frac{k}{m} \rceil$SUM with $O(n^{\delta})$ integers in each array, for any integer $m<k$ and $0<\delta \leq m$.
\end{theorem}

\begin{proof}
Given an instance of $k$SUM that contains $k$ arrays $A_1,A_2,...,A_k$ with $n$ integers in each of them and a target number $t$, we do the following (for now, we assume that $k$ is a multiple of $m$. Notice that $k$ and $m$ are considered as constants):
\begin{enumerate}

 \item Partition the $k$ arrays into $k/m$ groups of $m$ arrays. We denote the $i$th group in this partition by $G_i$.
 \item Pick an almost-linear almost-balanced hash function $h: [u] \rightarrow [n^{m-\delta}]$ and apply it to each element in every array ($[u]$ is some unbounded universe).
 \item For each possible choice of $t_1,t_2,...,t_{k/m-1} \in [n^{m-\delta}]$:
\begin{enumerate}[{3.}1]

 \item Find in each group $G_i$ all $m$-tuples $(a_{(i-1)m+1}^{j_1},a_{(i-1)m+2}^{j_2},...,a_{im}^{j_m})$, where $a_x^j$ is the $j$th element in $A_x$, such that $h(a_{(i-1)m+1}^{j_1}+a_{(i-1)m+2}^{j_2}+...+a_{im}^{j_m})= t_i$. We can find these $m$-tuples by solving $m$SUM reporting with group $G_i$ (after applying $h$) and the target number $t_i$. All $m$-tuples that are found are saved in a list $L_i$ ($L_i$ contains $m$-tuples that are found for a specific choice of $t_i$, after this choice is checked, as explained below, they are replaced by $m$-tuples that are found for a new choice of $t_i$).
 \item For $G_{k/m}$, find all $m$-tuples $(a_{k-m+1}^{j_1},a_{k-m+2}^{j_2},...,a_{k}^{j_m})$, such that \\ $h(a_{k-m+1}^{j_1}+a_{k-m+2}^{j_2}+...+a_{k}^{j_m})= t_{k/m}$. We can find these $m$-tuples by solving $m$SUM reporting with group $G_{k/m}$ (after applying $h$) and the target number $t_{k/m}$. All $m$-tuples that are found are saved in the list $L_{k/m}$. The value of the target number $t_{k/m}$ is fully determined by the values of $t_i$ we choose for the other groups, as the overall sum must be $h(t)$ in order for the original sum of elements to be $t$. Therefore, for $G_{k/m}$ the target value is $t_{k/m} = h(t)-\sum_{i=1}^{k/m-1}{t_i}$.
 \item For every $i \in [k/m]$, create an array $B_i$. For each $m$-tuple in $L_i$, add the sum of the elements of this tuple to $B_i$.
 \item Solve a $\frac{k}{m}$SUM instance with arrays $B_1,B_2,...,B_{k/m}$ and the target value $t$. If there is a solution to this $\frac{k}{m}$SUM instance return 1 - there is a solution to the original $k$SUM instance.
\end{enumerate}
\item Return 0 - there is no solution to the original $k$SUM instance.
\end{enumerate}

\textbf{Correctness}. If the original $k$SUM instance has a solution $a_1+ a_2+...+a_k=t$ such that $a_i \in A_i$ for all $i \in [k]$, then this solution can be partitioned to $k/m$ sums: $a_1+a_2+...+a_m=t'_1$, $a_{m+1}+a_{m+2}+...+a_{2m}=t'_2$,..., $a_{k-m+1}+a_{k-m+2}+...+a_k=t'_{k/m}$ for some integers $t'_1,t'_2,...t'_{k/m}$ such that $t'_{k/m} = t-\sum_{i=1}^{k/m-1}{t'_i}$. Therefore, by applying a hash function $h$, there is a solution to the original $k$SUM instance only if there are $t_1,t_2,...,t_{k/m-1} \in [n^{m-\delta}]$ such that: (a) $h(a_1+a_2+...+a_m)=t_1$, $h(a_{m+1}+a_{m+2}+...+a_{2m})=t_2$,..., $h(a_{k-m+1}+a_{k-m+2}+...+a_k)=t_{k/m}$ (b) For all $i$, $t_i = h(t'_i)$. This is exactly what is checked in step 3. However, as the hash function $h$ may cause false-positives (that is, we may have $t_1,t_2,...,t_{k/m-1} \in [n^{m-\delta}]$ such that their sum is $h(t)$ and $h(a_1+a_2+...+a_m)=t_1$, $h(a_{m+1}+a_{m+2}+...+a_{2m})=t_2$,..., $h(a_{k-m+1}+a_{k-m+2}+...+a_k)=t_{k/m}$, but $a_1+ a_2+...+a_k \neq t$), we need to verify each candidate solution. This is done in step (3.4).

The correctness of using $m$SUM (reporting) in steps (3.1) and (3.2) is due to the \emph{linearity} property of $h$ (see the note in Section~\ref{sec:preliminaries}). This linearity implies that finding all $m$-tuples in $G_i$ such that $h(a_{(i-1)m+1}^{j_1}+a_{(i-1)m+2}^{j_2}+...+a_{im}^{j_m})= t_i$ is equivalent to finding all $m$-tuples in $G_i$ such that $h(a_{(i-1)m+1}^{j_1})+h(a_{(i-1)m+2}^{j_2})+...+h(a_{im}^{j_m})= t_i$.

Regarding steps (3.3) and (3.4) we have the following observation:

\begin{observation}
The number of $m$-tuples that are saved in steps (3.1) and (3.2) in some $L_i$ for each possible value of $t_i$ is no more than $O(n^{\delta})$.
\end{observation}

The total number of $m$-tuples in some group $G_i$ is $n^m$. As $h$ is an almost-\emph{balanced} hash function (that can become balanced as explained in Appendix~\ref{apx:balanced}) with range $[n^{m-\delta}]$, the number of $m$-tuples that $h$ applied to the sum of their elements equals $t_i$ is expected to be at most $O(n^{\delta})$. However, this is true only if all these $m$-tuples have a different sum of elements. Unfortunately, there may be many $m$-tuples that the sum of their elements is equal, so all these $m$-tuples are mapped by $h$ to the same value $t_i$. Nevertheless, tuples with equal sum of elements are all the same for our purposes (we do not need duplicate elements in any $B_i$), as we are interested in the \emph{sum} of elements from all arrays no matter which specific elements sum up to it.

That being said, in steps (3.1) and (3.2) we do not add to $L_i$ every $m$-tuple that the sum of the elements of this tuple is $t_i$. Instead, for each $m$-tuple that $h$ over the sum of its elements equals $t_i$, we check if there is already a tuple with the same sum of elements in $L_i$ and only if there is no such tuple we add our $m$-tuple to $L_i$. In order to efficiently check for the existence of an $m$-tuple with the same sum in $L_i$, we can save the elements of $L_i$ in a balanced search tree or use some dynamic perfect hashing scheme. We call the process of removing $m$-tuples with same sum from $L_i$ the \textbf{removing duplicate sums} process.

The total number of $m$SUM and $\frac{k}{m}$SUM instances is determined by the number of possible choices for $t_1,t_2,...,t_{k/m-1}$ that is $O(n^{(k/m-1)(m-\delta)})$. Notice that $k$ and $m$ are fixed constants.

\medskip

\textbf{Modifications in the self-reduction for $k$ that is not a multiple of $m$}. In case $k$ is not a multiple of $m$, we partition the $k$ arrays into $\lceil k/m \rceil$ groups such that some of them have $m$ arrays and the others have $m-1$ arrays. In any case when we partition into groups of unequal size the range of the hash function $h$ is determined by the smallest group. If the smallest group has $d$ arrays then we use $h: [u] \rightarrow [n^{d-\delta}]$. Using this $h$ for groups of size $d$, we get all $d$-tuples that $h$ applied to their sum of elements equals some constant $t_i$. We expect $O(n^{\delta})$ such tuples (if we exclude d-tuples with the same sum as explained previously). However, for groups with more than $d$ arrays, say $d+\ell$, we expect the number of $(d+\ell)$-tuples that $h$ applied to their sum of elements equals $t_i$ to be $O(n^{\ell+\delta})$. Therefore, in order to just save all these tuples we must spend more space than we can afford to use. Therefore, we will only save $O(n^{\delta})$ of them in each time.

However, in order to be more efficient, we do not start solving $(d+\ell)$SUM reporting for every $O(n^{\delta})$ tuples we report on. Instead, we solve $(d+\ell)$SUM reporting once for all the expected $O(n^{\ell+\delta})$ $(d+\ell)$-tuples that $h$ applied to their sum of elements equals $t_i$. We do so by pausing the execution of $(d+\ell)$SUM reporting whenever we report on $O(n^{\delta})$ tuples. After handling the reported tuples we resume the execution of the paused $(d+\ell)$SUM reporting. We call this procedure of reporting on demand a partial output of the recursive calls, the \textbf{paused reporting} process.

As noted before, the number of $(d+\ell)$-tuples that $h$ applied to their sum of elements equals $t_i$ may be greater than $O(n^{\ell+\delta})$, because there can be many $(d+\ell)$-tuples that the sum of their elements is equal. We argued that we can handle this by saving only those tuples that the sum of their elements is unequal. However, in our case we save only $O(n^{\delta})$ tuples out of $O(n^{\ell+\delta})$ tuples, so we do not have enough space to make sure we do not save tuples that their sums were already handled. Nevertheless, the fact that we repeat handling tuples with the same sum of elements is not important since we anyway go over all possible tuples in our $(d+\ell)$SUM instance. The only crucial point is that in the last group that its target number is fixed, we have only $O(n^{\delta})$ elements for each $t_{\lceil k/m \rceil}$. This is indeed what happens if we take that group to be the group with the $d$ arrays (the smallest group). We call this important observation the \textbf{small space of fixed group} observation. That being said, our method can be applied even in case we partition to groups of unequal number of arrays.
\end{proof}

Using this self-reduction scheme we obtain the following Las Vegas solution to $k$SUM:

\setcounter{theorem}{4}

\begin{lemma}
For $k \geq 2$, $k$SUM can be solved by a Las Vegas algorithm following a self-reduction scheme that runs in $O(n^{k-\delta f(k)})$ time using $O(n^{\delta})$ space, for $0 \leq \delta \leq 1$ and $f(k)\ge \sum_{i=1}^{\log{\log{k}}}{k^{1/2^i}} - \log{\log{k}} - 2$.
\end{lemma}

\begin{proof}
Using our self-reduction from Theorem~\ref{thm:self_reduction}, we can reduce a single instance of $k$SUM to many instances of $m$SUM and $\frac{k}{m}$SUM for $m<k$. These instances can be solved recursively by applying the reduction many times.

\textbf{Solving the base case}. The base case of the recursion is $2$SUM that can be solved in the following way: Given two arrays $A_1$ and $A_2$, each containing $n$ numbers, our goal is to find all pairs of elements $(a_1,a_2)$ such that $a_1 \in A_1,a_2 \in A_2$ and $a_1+a_2=t$. This can be done easily in $\tilde{O}(n)$ time and $O(n)$ space by sorting $A_2$ and finding for each element in $A_1$ a matching element in $A_2$ using binary search. If we want to use only $O(n^{\delta})$ space we can do it by sorting only $O(n^{\delta})$ elements from $A_2$ each time and finding among the elements of $A_1$ a matching pair. This is done by scanning all elements of $A_1$ and binary searching the sorted portion of $A_2$. The total time for this procedure is $\tilde{O}(n^{2-\delta})$. Using hashing, following the general scheme we described previously, we can also obtain the same space-time tradeoff. We apply $h: [u] \rightarrow [n^{1-\delta}]$ to all elements of $A_1$ and $A_2$. For each value $t_1 \in [n^{1-\delta}]$ we find all elements $a_i \in A_1$ such that $h(a_i)=t_1$ and all elements $a_i \in A_2$ such that $h(a_i)=h(t)-t_1$. These elements form two arrays with $O(n^{\delta})$ elements in expectation, as we use an almost balanced hash function. These arrays serve as a 2SUM instance with $O(n^{\delta})$ elements, which we can solve by the regular (almost) linear time and linear space algorithm mentioned before. That being said, we get an $\tilde{O}(n^{2-\delta})$ time and $O(n^{\delta})$ space algorithm to solve 2SUM for any $0 \leq \delta \leq 1$.

We now analyse the running time of this solution. Denote by $T(k,n,s)$ the time needed to solve $k$SUM on input arrays of size $n$ with space usage at most $O(s)$. The full recursive process we have described to solve $k$SUM uses $O(n^{\delta})$ space (notice that the number of levels in the recursion depends only on $k$ that is considered constant) with the following running time: $T(k,n,n^{\delta}) = n^{(k/m-1)(m-\delta)}(T(m,n,n^{\delta})+T(k/m,n^{\delta},n^{\delta}))$. In order to solve the running time recursion, we start by solving it for the case that $\delta=1$. For this case we have that  $T(k,n,n) = n^{(k/m-1)(m-1)}(T(m,n,n)+T(k/m,n,n))$. The best running time in this case is obtained by balancing the two expressions within the parenthesis, which is done by setting $m=\sqrt{k}$. We have that $T(k,n,n)=2n^{(\sqrt{k}-1)(\sqrt{k}-1)}T(\sqrt{k},n,n) = 2n^{k-2\sqrt{k}+1}T(\sqrt{k},n,n)$. Solving this recursion we get that $T(k,n,n) = O(n^{k-\sum_{i=1}^{\log{\log{k}}}{k^{1/2^i}}+\log{\log{k}}})$.

Now, that we have solved the linear space case we can obtain a solution for any $\delta<1$ by plugging in this last result in our recursion. We have that $T(k,n,n^{\delta}) =$ \\ $n^{(k/m-1)(m-\delta)}(T(m,n,n^{\delta})+T(k/m,n^{\delta},n^{\delta})) =$ \\$n^{(k/m-1)(m-\delta)}(T(m,n,n^{\delta})+O(n^{\delta(k-\sum_{i=1}^{\log{\log{k}}}{k^{1/2^i}}+\log{\log{k}})}))$. It turns out that the best running time is obtained by setting $m=1$. For this value of $m$ we have that $T(k,n,n^{\delta}) = n^{(k-1)(1-\delta)}(T(1,n,n^{\delta})+O(n^{\delta(k-\sum_{i=1}^{\log{\log{k}}}{k^{1/2^i}}+\log{\log{k}})})) =$ \\ $n^{(k-1)(1-\delta)}(O(n)+O(n^{\delta(k-\sum_{i=1}^{\log{\log{k}}}{k^{1/2^i}}+\log{\log{k}})})) = O(n^{k-\delta\sum_{i=1}^{\log{\log{k}}}{k^{1/2^i}}+\delta(\log{\log{k}}+1)-1})$.
\end{proof}

Our self-reduction for the case $m=1$ becomes identical to the one presented by \\ Wang~\cite{Wang14}. However, the reduction by Wang is a reduction from $k$SUM to $k$SUM on a smaller input size, whereas our reduction is a general reduction from $k$SUM to $m$SUM for any $m<k$. Therefore, Wang has to present a different algorithm (discussed later in this paper) to solve $k$SUM using linear space. However, as this algorithm is Monte Carlo the whole solution is Monte Carlo. Using our generalized self-reduction we obtain a complete solution to $k$SUM. We have a linear space solution by choosing $m=\sqrt{k}$ and then we can use it to obtain a Las Vegas solution to $k$SUM for any $\delta \leq 1$ by choosing $m=1$.

\medskip

Regarding the self-reduction and its implications we should emphasize three points. The first one concerns our removing duplicate sums process. We emphasize that each time we remove a duplicate sum we regard to the \emph{original} values of the elements within that sum. An important point to observe is that duplicate sums that are caused by any hash function along the recursion, which are not duplicate sums according to the original values, do not affect the running time of our reduction. This is because the range of a hash function in a higher level of the recursion is larger than the total number of tuples we have in lower levels of the recursion. Thus, the number of duplicate sums that are caused by some hash function along the recursion is not expected to be more than $O(1)$. The second issue that we point out is the reporting version of $k$SUM and the output size. In our reduction in the top level of the recursion we solve $k$SUM without the need to report on all solutions. In all other levels of the recursion we have to report on all solutions (expect for duplicate sums). In our analysis we usually omit all references to the output size in the running time (and interchange between $k$SUM and its reporting variant). This is because the total running time that is required in order to report on all solutions is no more than $O(n^m)$ (for all levels of recursion), which does not affect the total running time as $m \leq k/2$. The third issue concerns rounding issues. In the proof of the general self-reduction we presented a general technique of how to handle the situation where $k$ is not a multiple of $m$. In order to make presentation clear we omit any further reference to this issue in the proof of the last lemma and the theorem in the next section. However, we emphasize that in the worst case the rounding issue may cause an increase by one in the exponent of the running time of the linear space algorithm. This is justified by the fact that the running time of the linear space algorithm is increased by one in the exponent or remains the same as we move from solving $k$SUM to solving $(k+1)$SUM. Moreover, the gap between two values of $k$, that the exponent of the running time does not change as we move from solving $k$SUM to $(k+1)$SUM, increases as a function of $k$. With that in mind, we decrease by one the lower bound on $f(k)$ and $g(k)$ in last lemma and the next theorem.

\medskip

In the following sections we present other benefits of our general self-reduction scheme.

\section{Improved Deterministic Solution for $k$SUM}
Using the techniques of~\cite{LWWW16} our randomized solution can be transformed to a deterministic one by imitating the hash function behaviour in a deterministic way. This way we get the following result:

\setcounter{theorem}{1}

\begin{theorem}
For $k \geq 2$, $k$SUM can be solved by a deterministic algorithm that runs in $O(n^{k-\delta g(k)})$ using $O(n^{\delta})$ space, for $0 \leq \delta \leq 1$ and $g(k)\ge \sqrt{k} - 2$.
\end{theorem}

\begin{proof}
We partition the $k$ arrays into $k/m$ groups of $m$ arrays. We denote the $i$th group in this partition by $G_i$. For every group $G_i$, there are $n^m$ sums of $m$ elements, such that each element is from a different array of the $m$ arrays in $G_i$. We denote by $SUMS_{G_i}$ the array that contains all these sums. A \emph{sorted part} of a group $G_i$ is a continuous portion of the sorted version of $SUMS_{G_i}$. The main idea for imitating the hash function behaviour in a deterministic way is to focus on sorted parts of size $n^{\delta}$, one for each of the first $k/m-1$ groups. Then the elements from the last group that are candidates to complete the sum to the target number are fully determined. Each time we pick different $n^{\delta}$ elements out of these elements and form an instance of ($\frac{k}{m}$)SUM such that the size of each array is $n^{\delta}$. The crucial point is that the total number of these instances will be $O(n^{(k/m-1)(m-\delta)})$ as in the solution that uses hashing techniques. This is proven based on the domination lemma of~\cite{LWWW16} (see the full details in Section 3.1 of~\cite{LWWW16}). Lincoln et al.~\cite{LWWW16} present a corollary of the domination lemma as follows: \emph{Given a $k$SUM instance $L$, suppose $L$ is divided into $g$ groups $L_1,..., L_g$ where $|L_i| = n/g$ for all $i$, and for all $a \in L_i$ and $b \in L_{i+1}$ we have $a \leq b$. Then there are $O(k \cdot g^{k-1})$ subproblems $L'$ of $L$ such that the smallest $k$SUM of $L'$ is less than zero and the largest $k$SUM of $L'$ is greater than zero.} Following our scheme, $g$ in this corollary equals $n^{m-\delta}$ in our case (there are $g$ groups of size $n^{\delta}$ in each $SUMS_{G_i}$) and the $k$ in the corollary is in fact $k/m$ in our case. Therefore, we get that the total number of instances that have to be checked is indeed $O(n^{(k/m-1)(m-\delta)})$.

In order for this idea to work, we need to obtain a sorted part of size $n^{\delta}$ from each group $G_i$. In this case, we do not have the recursive structure as in the randomized solution because we no longer seek for $m$ elements in each group that sum up to some target number, but rather we would like to get a sorted part of each group. However, we can still gain from the fact that we have only $O(n^{(k/m-1)(m-\delta)})$ instances of ($\frac{k}{m}$)SUM.

Lincoln et al.~\cite{LWWW16} presented a simple data structure that obtains a sorted part of size $O(S)$ from an array with $n$ elements using $O(n)$ time and $O(S)$ space. We can use this data structure in order to obtain a sorted part of $n^{\delta}$ elements for each group $G_i$ by considering the elements of the array $SUMS_{G_i}$. Consequently, a sorted part of $n^{\delta}$ elements from $G_i$ can be obtained using $O(n^m)$ time and $O(n^{\delta})$ space.

Putting all parts together we have a deterministic algorithm that solves $k$SUM with the following running time: $T(k,n,n^{\delta}) = n^{(k/m-1)(m-\delta)}(n^m+T(k/m,n^{\delta},n^{\delta}))$. By setting $m=\sqrt{k}$ we have $T(k,n,n^{\delta}) = n^{k-\sqrt{k}-\sqrt{k}\delta+\delta}(n^{\sqrt{k}}+T(\sqrt{k},n^{\delta},n^{\delta}))$. Solving $k$SUM using linear space can be trivially done using $n^k$ time. Therefore, we get that $T(k,n,n^{\delta}) = n^{k-\sqrt{k}-\sqrt{k}\delta+\delta}(n^{\sqrt{k}}+n^{\delta\sqrt{k}}) = n^{k-\sqrt{k}\delta+\delta}$.
\end{proof}

The last theorem is a significant improvement over the previous results of Lincoln et al.~\cite{LWWW16} that obtain just a small improvement of at most 3 in the exponent over the trivial solution that uses $n^k$ time, whereas our solution obtains an improvement of almost $\sqrt{k}\delta$ in the exponent over the trivial solution.

\section{Las Vegas Variant of Wang's Linear Space Algorithm}\label{sec:las_vegas_wang}
Wang~\cite{Wang14} presented a Monte Carlo algorithm that solves $(T_j+1)$SUM in $O(n^{T_{j-1}+1})$ time and linear space, where $T_j = \sum_{i=1}^{j}{i}$. We briefly sketch his solution here in order to explain how to modify it in order to obtain a Las Vegas algorithm instead of a Monte Carlo algorithm. Given an instance of $k$SUM with $k$ arrays $A_1,A_2,...,A_k$ such that $k=T_j+1$, he partitions the arrays into two groups. The left group contains the first $j$ arrays and the right group all the other arrays. An almost-linear almost-balanced hash function $h$ is chosen, such that its range is $m' = \Theta(n^{j-1})$. The hash function $h$ is applied to all elements in all input arrays. Then, the algorithm goes over all possible values $v_l \in [m']$. For each such value, the first array of the left group is sorted and for all possible sums of elements from the other $j-1$ arrays (one element from each array) it is checked (using binary search) if there is an element from the first array that completes this sum to $v_l$. If there are $j$ elements that their hashed values sum up to $v_l$ they (the original values) are saved in a lookup table $T$. At most $\Theta(n)$ entries are stored in $T$. After handling the left group the right group is handled. Specifically, if the target value is $t$ the sum of elements from the arrays in the right group should be $h(t)-v_l$ (to be more accurate as our hash function is almost linear we have to check $O(1)$ possible values). To find the ($k-j$)-tuples from the right group that sum up to $h(t)-v_l$ a recursive call is done on the arrays of the right group (using their hashed version) where the target value is $h(t)-v_l$. A crucial point is that the algorithm allows the recursive call to return at most $n^{T_{j-2}+1}$ answers. For each answer that we get back from the recursion, we check, in the lookup table $T$, if the \emph{original} values of the elements in this answer can be completed to a solution that sums up to the target value $t$. The number of answers the algorithm returns is at most $num$ which in this case is $n^{T_{j-1}+1}$. If there are more answers than $num$ the algorithm returns (to the previous level in the recursion).

In order for this algorithm to work, Wang uses a preliminary Monte Carlo procedure that given an instance of $k$SUM creates $O(\log{n})$ instances of $k$SUM such that if the original instance has no solution none of these instances has a solution and if it has a solution at least one of these instances has a solution but no more than $O(1)$ solutions. The guarantee that there are at most $O(1)$ solutions is needed to ensure that each recursive call is expected to return the right number of solutions. For example, if the algorithm does a recursive call as explained before on $k-j=T_j+1-j=T_{j-1}+1$ arrays, then we expect that for each value of $h(t)-v_l$ out of the $\Theta(n^{j-1})$ possible values, at most $O((T_{j-1}+1)/n^{j-1})=O(n^{T_{j-2}+1})$ answers will be returned from the recursive call. This is because of the almost-balanced property of the hash function. However, if there are many ($k-j$)-tuples whose sum is equal (in their original values), then they will be mapped to the same hash value due to the linearity property of the hash function (to be more accurate, as our hash function is almost-linear there are $O(1)$ possible values that these elements can be mapped to). This is where Wang uses the fact that the new instance of $k$SUM has no more than $O(1)$ solutions. The number of answers that is returned from the recursive call can be limited to the expected value, as there are at most $O(1)$ ($k-j$)-tuples that have equal sum and are part of a solution because each one of these sums forms a different solution to our $k$SUM instance and there are at most $O(1)$ such solutions.

We now explain how to modify this algorithm in order to make it a Las Vegas algorithm. The idea is to use the tools we presented for our general self-reduction. This is done in the following theorem:

\begin{theorem}
For $k \geq 2$, $k$SUM can be solved by a Las Vegas algorithm that runs in $O(n^{k-\delta f(k)})$ time using $O(n^{\delta})$ space, for $0 \leq \delta \leq 1$ and $f(k)\ge \sqrt{2k} - 2$
\end{theorem}

\begin{proof}
We begin with the algorithm by Wang. The first modification to the algorithm is not to limit the number of answers returned from the recursive call. Let us look at some point in the recursion for which we have $j$ arrays in the left group and $k'-j$ in the right group where the total number of arrays is $k'=T_{j}+1$. Wang limited the total number of answers we receive from each of the $n^{j-1}$ recursive calls to be $n^{k'-j}/n^{j-1}=n^{T_{j}-j+1}/n^{j-1}=n^{T_{j-2}+1}$. This is the expected number of answers we expect to get using a balanced hash function where we do not expect to have many duplicate identical sums. However, even if we do not limit the number of answers we get back from a recursive call the total number of answers we receive back from all the $n^{j-1}$ recursive calls is at most $n^{T_{j-1}+1}$. This is simply because the number of arrays in the right group is $T_{{j-1}+1}$. As there can be duplicate sums in this right group the number of answers that we receive from each recursive call (out of the $\Theta(n^{j-1})$ recursive calls) can be much larger than the number of answers we get from another recursive call. Nevertheless, the total number of answers is bounded by the same number as in Wang's algorithm. Now, considering the left group, for every possible value of $v_l \in \Theta(n^{j-1})$ we expect the number of $j$-tuples that are their hashed sum is $v_l$ to be $O(n)$. This is true unless we have many equal sums that, as explained before, are all mapped to the same value by $h$. In order to ensure that we save only $O(n)$ $j$-tuples in the lookup table $T$, we use our "removing duplicate sums" process. That is, for each $j$-tuple that is mapped by $h$ to some specific $v_l$ we ensure that there is no previous $j$-tuple in $T$ that has the same sum (considering the original values of the elements).

Following this modification of the algorithm, we have that the left group is balanced as we expect no more than $O(n)$ entries in $T$ for each possible value of $v_l$, while the right group may not be balanced. However, what is important is that one group is balanced and the total number of potential solutions in the other groups is the same as in the balanced case. Therefore, we can apply here our "small space of fixed group" observation (see Section~\ref{sec:self_reduction}) that guarantees the desired running time. Verifying each of the answers we get from the right group can be done using our lookup table in $O(1)$ time. Since we have removed duplicate sums (using original values) the expected number of elements that can complete an answer from the right group to a solution to the original $k$SUM instance is no more than $O(1)$. This is because the number of elements mapped to some specific value of $h$ and having the same value by some $h'$ from some upper level of our recursion is not expected to be more than $O(1)$, as the range of $h'$ is at least $n^{j}$ and the number of $j$-tuples is $n^j$.  Therefore, the total running time will be $O(n^{T_{j-1}+1})$ even for our modified algorithm. Moreover, the expected number of answers that are returned by the algorithm for a specific target value is $O(n^{T_{j-2}+1})$.

We note that the answers that are returned from the right group are returned following the "paused reporting" scheme we have described in our self-reduction. We get answers one by one by going back and forth in our recursion and pausing the execution each time we get a candidate solution (it seems that it is also needed in Wang's algorithm though it was not explicitly mentioned in his description).

To conclude, by modifying Wang's algorithm so that the number of the answers returned to the previous level of recursion is not limited and by removing duplicates in the right group (within every level of recursion) we obtained a Las Vegas algorithm that solves $(T_{j}+1)$SUM in $O(n^{T_{j-1}+1})$ time and linear space. Using the self-reduction with $m=1$ we have a Las Vegas algorithm that solves $k$SUM using $O(n^{k-\delta f(k)})$ time, for $f(k) \ge \sqrt{2k} - 2$, and $O(n^{\delta})$ space, for $0 \leq \delta \leq 1$.
\end{proof}

This Las Vegas algorithm has a better running time than an algorithm using the self-reduction directly because of the additional $\sqrt{2}$ factor before the $-\sqrt{k}$ in the exponent. However, as we will explain in the following sections, there are other uses of our general self-reduction approach.

\section{Space-Time Tradeoffs for Large Space}
We now consider how to solve $k$SUM for the case where we can use space which is $O(n^{\delta})$ for $\delta>1$. We have two approaches to handle this case. The first one is a generalization of the Las Vegas algorithm from the previous section. The second uses our general self-reduction approach from Section~\ref{sec:self_reduction}.

We begin with the first solution and obtain the following result:

\setcounter{theorem}{5}
\begin{lemma}
For $k \geq 2$, $k$SUM can be solved by a Las Vegas algorithm that runs in $O(n^{k-\sqrt{\delta} f(k)})$ time using $O(n^{\delta})$ space, for integer $\frac{k}{4} \geq \delta > 1$ and $f(k)\ge \sqrt{2k} - 2$.
\end{lemma}

\begin{proof}
As an input we have an instance of $k$SUM with arrays $A_1,A_2,...,A_k$ such that $k=T^{\delta}_j+\delta$ (for other values of $k$ see the details below), where $T^{\delta}_j =\sum_{i=1}^{j}{\delta i}$. The target number is $t$. We partition the arrays into two groups. The left group contains the first $\delta j$ arrays and the right group all the other arrays. An almost-linear almost-balanced hash function $h$ is chosen, such that its range is $m' = \theta(n^{\delta{j-1}})$. The hash function  $h$ is applied to all elements in all input arrays. We go over all possible values $v_l \in [m']$. For each such value, we compute all hashed sums of $\delta$-tuples from the first $\delta$ arrays (each element in the sum is taken from a different array) of the left group and we sort them. We note that all duplicate sums (according to their original values) are removed. Then, for all possible hashed sums of elements from the other $\delta(j-1)$ arrays (one element from each array) we check (using binary search) if there is a hashed $\delta$-tuple from the first $\delta$ arrays that completes this hashed sum to $v_l$. If there are $\delta j$ elements such that their hashed values sum up to $v_l$ they, the original values, are saved in a lookup table $T$. Again, sums that are identical according to original values of elements are removed. At most $\theta(n^{\delta})$ entries are expected to be stored in $T$. After handling the left group the right group is handled. Specifically, if the target value is $t$ the sum of elements from the arrays in the right group should be $h(t)-v_l$ (to be more accurate as our hash function is almost linear we have to check $O(1)$ possible values). To find the ($k-\delta j$)-tuples from the right group that sum up to $h(t)-v_l$ a recursive call is done on the arrays of the right group (using their hashed version) where the target value is $h(t)-v_l$. For each answer that we get back from the recursion, we check in the lookup table $T$ if the original values of the elements in this answer can be completed to a solution that sum up to the target value $t$.

The base case of the recursion is $k=2\delta$. In this case we can split the arrays into two groups of $\delta$ arrays and compute all possible sums of $\delta$ elements (one from each of the arrays). Then, having two lookup tables containing the $n^{\delta}$ possible sums of each group, we can check in $n^{\delta}$ time if there is a pair of $\delta$ tuples one from each group such that their sum is $t$.

The running time analysis is similar to the case of $\delta$ equals $1$. Specifically, it can be shown by induction that the algorithm solves $(T^{\delta}_{j}+\delta)$SUM in $O(n^{T^{\delta}_{j-1}+\delta})$ time using $O(n^{\delta})$ space. For j=1, we have that the base case can be solved in $O(n^{\delta})$ time as required. Let us assume that the running time is correct for $j-1$ and prove it for $j$. Given an instance of $k$SUM with $k=T^{\delta}_{j}+\delta$ the algorithm goes over all possible values of $v_l$ which are $O(n^{\delta{j-1}})$. For each possible value the work that is done on the left group is $O(n^{\delta(j-1)})$. On the right group we do a recursive call on $T^{\delta}_{j-1}+\delta$ arrays that runs according to the induction hypothesis in $O(n^{T^{\delta}_{j-2}+\delta})$ time. Each returned answer is checked in constant time against the lookup table of the left group in order to verify which of the $O(1)$ possible values in the left group (if there is any) that can complete it to the target value. As $\delta(j-1)\leq T^{\delta}_{j-2}+\delta$ for $j>1$, we get that the total running time is $O(n^{\delta(j-1)}n^{T^{\delta}_{j-2}+\delta})=O(n^{T^{\delta}_{j-1}+\delta})$ as required.

As usual, for values of $k$ that are not equal to any $T^{\delta}_{j}+\delta$, we can solve $k$SUM by running $n$ instances of $(k-1)$SUM.

In this solution we have that $k=T^{\delta}_{j}+\delta=\delta T_{j} +\delta = \delta (j+1)j/2 +\delta$. Therefore, the value of $j$ is approximately $\sqrt{2k/\delta}$. Following our previous analysis in order to solve $k$SUM the running time is $O(n^{T^{\delta}_{j-1}+\delta})= O(n^{\delta(T_{j-1}+1)})=O(n^{k-\delta j})$. Consequently the running time to solve $k$SUM using $O(n^{\delta})$ space is $O(n^{k-\sqrt{\delta}f(k)})$ for $f(k) \ge \sqrt{2k} - 2$.
\end{proof}

The approach of the last theorem has one drawback - it gives no solution for the case where we can use $O(n^{\delta})$ space for non integer $\delta>1$. To solve this case we use our general self-reduction approach and obtain the following:

\setcounter{theorem}{3}
\begin{theorem}
For $k \geq 2$, $k$SUM can be solved by a Las Vegas algorithm that runs in $O(n^{k-\sqrt{\delta} f(k)})$ time using $O(n^{\delta})$ space, for $\frac{k}{4} \geq \delta > 1$ and $f(k)\ge \sqrt{2k} - 2$.
\end{theorem}

\begin{proof}
Recall that the idea of the self-reduction is to split our $k$SUM instance into $k/m$ groups of $m$ arrays. An almost-linear almost-balanced hash function $h$ is applied to all elements. Then, each group is solved recursively and the answers reported by all of these $k/m$ $m$SUM instances form an instance of $(\frac{k}{m})$SUM. This approach leads to the following recursive runtime formula: $T(k,n,n^{\delta}) = n^{(k/m-1)(m-\delta)}(T(m,n,n^{\delta})+T(k/m,n^{\delta},n^{\delta}))$ (see the full details in Section~\ref{sec:self_reduction}). This approach works even for $\delta>1$. It turns out that the best choice of $m$ for $\delta>1$ is $m=\lceil \delta \rceil$, which coincides with our choice of $m$ for $\delta \leq 1$. Following this choice, we have that $T(k,n,n^{\delta}) = n^{(k/\lceil \delta \rceil-1)(\lceil \delta \rceil-\delta)}(T(\lceil \delta \rceil,n,n^{\delta})+T(k/\lceil \delta \rceil,n^{\delta},n^{\delta}))=O(n^{(k/\lceil \delta \rceil-1)(\lceil \delta \rceil-\delta)}(n^{\lceil \delta \rceil}+T(k/\lceil \delta \rceil,n^{\delta},n^{\delta})))$. If we plug in our Las Vegas solution following Wang's approach from Section~\ref{sec:las_vegas_wang}, we get that the running time is approximately $O(n^{(k/\lceil \delta \rceil-1)(\lceil \delta \rceil-\delta)}(n^{\lceil \delta \rceil}+n^{\delta(k/\lceil \delta \rceil) - \delta\sqrt{2k/\lceil \delta \rceil}}))$. Therefore, the running time to solve $k$SUM using $O(n^{\delta})$ space for $\delta>1$ is $O(n^{k-\sqrt{\delta}f(k)})$ for $f(k) \ge \sqrt{2k} - 2$.
\end{proof}

We see that for integer values of $\delta$ the last result coincides with the previous approach. Using the self-reduction approach even for non integer values of $\delta$, we get similar running time behaviour.
We note that using the same ideas from Theorem 3 and the results from Section 4 we can have the same result as Theorem 3 for a deterministic algorithm but with running time which is $O(n^{k-\sqrt{\delta}f(k)})$ for $f(k) \ge \sqrt{k} - 2$.

\section{Space Efficient Solutions to $6$SUM}
In this section we present some space efficient solutions to $6$SUM that demonstrate the usefulness of our general self-reduction in concrete cases. For $3$SUM we do not know any truly subquadratic time solution and for $4$SUM we have an $O(n^2)$ time solution using linear space, which seems to be optimal. Investigating $6$SUM is interesting because in some sense $6$SUM can be viewed as a problem that has some of the flavour of both $3$SUM and $4$SUM, which is related to the fact that $2$ and $3$ are the factors of $6$. Specifically, $6$SUM has a trivial solution running in $O(n^3)$ time using $O(n^3)$ space. However, when only $O(n)$ space is allowed $6$SUM can be solved in $O(n^4)$ time by Wang's algorithm. More generally, using Wang's solution $6$SUM can be solved using $O(n^{\delta})$ space in $O(n^{5-\delta})$ time for any $\delta \leq 1$. As one can see, on the one hand $6$SUM can be solved in $O(n^3)$ time that seems to be optimal, which is similar to $4$SUM. On the other hand, when using at most linear space no $O(n^{4-\epsilon})$ solution is known for any $\epsilon>0$, which has some flavour of the hardness of $3$SUM.

There are two interesting questions following this situation:
(i) Can $6$SUM be solved in $O(n^3)$ time using less space than $O(n^3)$?
(ii) Can $6$SUM be solved in $O(n^4)$ time using truly sublinear space?. Using our techniques we provide a positive answer to both questions.

We begin with an algorithm that answer the first question and obtain the following result:

\setcounter{theorem}{6}

\begin{theorem}
There is a Las Vegas algorithm that solves $6$SUM and runs in $O(n^{5-\delta}+n^3)$ time using $O(n^{\delta})$ space, for any $\delta \geq 0.5$.
\end{theorem}

\begin{proof}
Given an instance of $6$SUM with arrays $A_1,A_2,...,A_6$ and a target number $t$, we split the arrays into two groups $G_1$ with arrays $A_1,A_2,A_3$ and $G_2$ with arrays $A_4,A_5,A_6$. We pick an almost-balanced almost-linear hash function $h$ that its range is $n^{3-\delta}$ for some $\delta \geq 0$. Then, we apply $h$ to all elements in all input arrays. For each value $t_1 \in [n^{3-\delta}]$, we have two instances of 3SUM to solve. The first contains the arrays in $G_1$ and $t_1$ as a target value, and the second contains the arrays in $G_2$ and $h(t)-t_1$ as the target value. We adapt the 3SUM algorithm such that it returns all 3-tuples that sum up to the target value and not just return if there is such tuple or not. This can be done in $O(n^2 + op)$ time, where $op$ is the output size, using linear space or even $\tilde{O}(\sqrt{n})$ space by using the algorithm of Wang~\cite{Wang14}. From the returned output we remove all duplicate sums, so the expected size of the output (after removing duplicates) is $O(n^{\delta})$ as $h$ is an almost-balanced hash function. Therefore, for each value of $t_1$ we have two arrays that each of them contains $O(n^{\delta})$ elements. These two arrays (when considering the sums of the returned tuples) form an instance of $2$SUM that its target is $t$. If there is a solution to this instance then we have a solution to our $6$SUM instance, if not we continue to check the next value of $t_1$. If no solution is found for any of the $n^{3-\delta}$ instances of $2$SUM then there is no solution to the original $6$SUM instance.

It easy to verify that the algorithm uses $O(n^{\delta})$ space for $\delta>0.5$. Regarding the running time, there are $n^{3-\delta}$ possible values of $t_1$ and for each one of them we solve two instances of $3$SUM that contain $O(n)$ elements, and one instance of $2$SUM that contains $O(n^{\delta})$ elements. Therefore, the total running time of this solution is $O(n^{3-\delta}(n^2+n^{\delta}))=O(n^{5-\delta}+n^3)$ for $\delta>0.5$ (we note that the total output size from all the $3$SUM instances is $O(n^3)$ which adds just a constant factor to the running time).
\end{proof}

By the last theorem we get a tradeoff between time and space which demonstrates in one extreme that $6$SUM can be solved in $O(n^3)$ time using $O(n^2)$ space instead of the $O(n^3)$ space of the trivial solution.

The algorithm from the previous theorem runs in $O(n^4)$ time while using $O(n)$ space, this is exactly the complexity of Wang's algorithm for $6$SUM. We now present an algorithm that runs in $O(n^4)$ time but uses truly sublinear space.

\begin{theorem}
There is a Las Vegas algorithm that solves $6$SUM and runs in $O(n^{6-3\delta}+n^4)$ time using $O(n^{\delta})$ space, for any $\delta \geq 0$.
\end{theorem}

\begin{proof}
As an input the algorithm is given an instance of $6$SUM with arrays $A_1,A_2,...,A_6$ and a target number $t$. We split the arrays into three groups $G'_1$ with arrays $A_1,A_2$, $G'_2$ with arrays $A_3,A_4$ and $G'_3$ with arrays $A_5,A_6$. We pick an almost-balanced almost-linear hash function $h$ with a range that is $n^{2-\delta}$ for some $\delta \geq 0$. The hash function $h$ is applied to all elements in all arrays. For each value of $t_1,t_2 \in [n^{2-\delta}]$ we create three instances of $2$SUM. Each instance is composed of a different group from $G'_1,G'_2$ and $G'_3$, while the target number is $t_1,t_2$ and $h(t)-t_1-t_2$, respectively. $2$SUM can be solved using $O(n^{\delta})$ space in $O(n^{2-\delta})$ time as explained in Section~\ref{sec:self_reduction}. We note that all duplicate sums are removed and the \emph{total} output size for \emph{all} $2$SUM instances is $O(n^2)$, so returning the output does not add a significant cost to our total running time. Each of the $2$SUM instances is expected to return $O(n^{\delta})$ answers (3-tuples) because of the almost-balanced property of $h$ (after removing duplicate sums). We create three arrays, one for each of the $O(n^{\delta})$ 3-tuples, that contain the sums of these tuples. These arrays form an instance of $3$SUM with $O(n^{\delta})$ elements and the target value $t$. If any of the $3$SUM instances for all possible values of $t_1$ and $t_2$ returns a positive answer then we have a solution, otherwise there is no solution to our $6$SUM instance.

We analyse the running time of the algorithm. For each of the $n^{4-2\delta}$ possible values of $t_1$ and $t_2$, the algorithm solve three instances of $2$SUM with $O(n)$ elements (with space usage limited to $O(n^{\delta})$) and a single instance of $3$SUM with $O(n^{\delta})$ elements. Therefore, the total running time is $O(n^{4-2\delta}(n^{2-\delta}+n^{2\delta}))=O(n^{6-3\delta}+n^4)$. The total space complexity is $O(n^{\delta})$.
\end{proof}

By setting $\delta=2/3$ in the last theorem, we have an algorithm that solves $6$SUM in $O(n^4)$ time while using just $O(n^{2/3})$ space.

\section{Improved Space-Efficient Solution to 3SUM on Random Input}\label{sec:3sum_random_input}
In this section we present another application of our general self-reduction method. Recent research work by Bansal et al.~\cite{BGNV17} with the motivation of providing a space efficient solution to SUBSET-SUM, presented an efficient algorithm to the List Disjointness problem using small space. The List Disjointness problem is closely related to $2$SUM, so actually their algorithm also provide an efficient solution to $2$SUM using small space. However, it is worth noting that their solution is based on the strong assumption that we have a random read-only access to polynomially many random bits. With their algorithm in mind, it is interesting to investigate the implication of an improved solution to $2$SUM on the complexity of $k$SUM for other values of $k$, where the case of $k=3$ is especially interesting.

Let us assume that there is an algorithm that solves $2$SUM using $O(n^{\delta})$ space ($0 \leq \delta \leq 1$) and runs in $O(n^{f(\delta)})$ time for some $f$ which is a function of $\delta$. Moreover, let us assume that this algorithm can be adapted to solve the reporting version of $2$SUM

We show how to solve $3$SUM using our techniques, such that an improved solution to $2$SUM implies an improved solution to $3$SUM. This is done in the following theorem:

\setcounter{theorem}{8}

\begin{theorem}
Assuming there is an algorithm that solve $2$SUM reporting using $O(n^{\delta})$ space in $O(n^{f(\delta)})$ time, $3$SUM can be solved using $O(n^{\delta})$ space and $O(n^{1-\delta+f(\delta)}+n^2)$ time.
\end{theorem}

\begin{proof}
As an input the algorithm is given an instance of $3$SUM with arrays $A_1,A_2,A_3$ and a target number $t$. We split the input into two groups $G_1$ that contains $A_1$ and $A_2$ and $G_2$ that contains $A_3$. As usual, we pick an almost-balanced almost-linear hash function $h$, but this time its range is $n^{1-\delta}$. The hash function $h$ is applied to all elements in all arrays. For each possible value of $t_1 \in [n^{1-\delta}]$, we can get all the elements in $G_2$ that their hashed value equals $h(t)-t_1$. The expected number of such elements is $O(n^{\delta})$ because $h$ is almost balanced. Note that duplicate elements are removed. Group $G_1$ forms an instance of $2$SUM with the target number $t_1$. We use the algorithm solving $2$SUM in order to retrieve all pairs of elements (one from $A_1$ and the other form $A_2$) that their hashed sum is $t_1$. At first glance the expected number of pairs returned is $O(n^{1+\delta})$ as $h$ is an almost-balanced hash function, but as there can be many pairs of elements that their sum (using the original values) is the equal the number of returned pairs can be much larger. However, we can use our "paused reporting" technique from Section~\ref{sec:self_reduction} in order to get just $O(n^{\delta})$ portion of the output each time. As the total number of answers returned by all $2$SUM instances, for the different values of $t_1$, is $O(n^2)$ no matter how many duplicate sums we have, the number of times we have to run the $2$SUM algorithm is $O(n^{1-\delta})$ anyway. What is important following our "small space of fixed group" observation from Section~\ref{sec:self_reduction} is that the elements we have from group $G_2$ is limited in expectation to $O(n^{\delta})$ after removing duplicates (which can be done for $G_2$). That being said, for each value of $t_1$ we have $O(n^{\delta})$ elements from group $G_2$ that their hashed value is $h(t)-t_1$, which we save in array $B_1$, and we can also get $O(n^{\delta})$ pairs of elements from group $G_1$ that their hashed sum is $t_1$. For these pairs, we save in array $B_2$ the sum of elements within the pair. The arrays $B_1$ and $B_2$ forms an instance of $2$SUM with target value $t$. There is a solution to our $3$SUM instance iff there is a solution to one of the $O(n^{2-\delta})$ $2$SUM instances on $O(n^{\delta})$ elements.

The space complexity of the algorithm is $O(n^{\delta})$, as each of the $2$SUM instances we solve does not consume more than $O(n^{\delta})$ space. Regarding the running time we run $O(n^{1-\delta})$ instances of $2$SUM with $O(n)$ elements (using the paused reporting technique) and another $O(n^{2-\delta})$ instances of $2$SUM with $O(n^{\delta})$ elements. Therefore the total running time to solve $3$SUM using $O(n^{\delta})$ space is $O(n^{1-\delta}n^{f(\delta)} + n^{2-\delta}n^\delta)=O(n^{1-\delta+f(\delta)}+n^2)$.
\end{proof}

The same ideas can be used to obtain the following results for $4$SUM, $6$SUM and $8$SUM:

\begin{theorem}
Assuming there is an algorithm that solve $2$SUM reporting using $O(n^{\delta})$ space in $O(n^{f(\delta)})$ time, $4$SUM, $6$SUM and $8$SUM can be solved using $O(n^{\delta})$ space and $O(n^{2-\delta+f(\delta)}+n^2)$, $O(n^{4-2\delta+f(\delta)}+n^4)$ and $O(n^{6-3\delta+f(\delta)}+n^{6-\delta})$ time, respectively.
\end{theorem}

\begin{proof}
For $4$SUM, we partition the 4 input arrays into two groups of two arrays. We apply an almost-linear almost-balanced hash function hash function $h$ that its range is $n^{2-\delta}$ to all elements. For each value of $t_1 \in [n^{2-\delta}]$ we solve two instances of $2$SUM, each one of them for a different group with different target values $t_1$ and $h(t)-t_1$. We get $O(n^{2-\delta})$ answers in expectation (after removing duplicate sums) from each of these instances. Then, we solve a $2$SUM instance on these two arrays of $O(n^{\delta})$ elements. There is a solution to the $4$SUM instance iff there is a solution to at least one of the $2$SUM instances on $O(n^{\delta})$ elements. The space complexity is obviously $O(n^{\delta})$. The total running time is $O(n^{2-\delta}(n^{f(\delta)}+n^{\delta}))=O(n^{2-\delta+f(\delta)}+n^2)$.

For $6$SUM, we split the input arrays to 3 groups of 2 arrays. We apply the same $h$ as for $4$SUM to all elements. For each value of $t_1,t_2 \in [n^{2-\delta}]$ we have three $2$SUM instances to solve. On the output of these instances we run $3$SUM algorithm.  There is a solution to the $6$SUM instance iff there is a solution to at least one of the $3$SUM instances. The space complexity is obviously $O(n^{\delta})$. The total running time is $O(n^{4-2\delta}(n^{f(\delta)}+n^{2\delta}))=O(n^{4-2\delta+f(\delta)}+n^4)$.

For $8$SUM, we split the input arrays to 4 groups of 2 arrays. Again, we apply the same $h$ as for $4$SUM to all elements. For each value of $t_1,t_2,t_3 \in [n^{2-\delta}]$ we have four $2$SUM instances to solve. On the output of these instances we run $4$SUM algorithm.  There is a solution to the $8$SUM instance iff there is a solution to at least one of the $4$SUM instances. The space complexity is $O(n^{\delta})$. The total running time is $O(n^{6-3\delta}(n^{f(\delta)}+n^{2\delta}))=O(n^{6-3\delta+f(\delta)}+n^{6-\delta})$.
\end{proof}

In the same way similar results can be obtained for other values of $k$.

Bansal et al.~\cite{BGNV17} presented an randomized algorithm that solves List Distinctness and therefore also $2$SUM in $\tilde{O}(n\sqrt{p/s})$ time using $O(s\log{n})$ space, under the assumption that there is a random read-only access to a random function $h : [m] \rightarrow [n]$ ($m$ is the universe of the elements). The value of $p$ measures the number of \emph{pseudo-solutions} in both input arrays, informally speaking, it is a measure of how many pairs of indices there are in each array that have the same value (see~\cite{BGNV17} for full details). In case we have a random input then $p=\Theta(n)$, if the universe size is at least $n$. Using this algorithm Bansal et al. proved that $k$SUM on random input can be solved using $\tilde{O}(n^{k-0.5})$ time and $O(\log{n})$ space, for $k \geq 2$. For even $k$ under the assumption that there is a random read-only access to random bits, they present a solution to $k$SUM on random input that runs in $\tilde{O}(n^{3k/4}\sqrt{s})$ time and $O(s\log{n})$ space.

The algorithm of Bansal et al.~\cite{BGNV17} that solves List Distinctness and also $2$SUM can be adapted in a straightforward manner to an algorithm that reports all solutions in $\tilde{O}(n\sqrt{p/s}+op)$ time (where $op$ is the size of the output) using $O(s\log{n})$ space, under the assumption that there is a random read-only access to a random function $h : [m] \rightarrow [n]$. The intuition behind this is that if we have many answers to report on then the task of reporting on the first ones is easier than finding just one solution, in case it is the only one that exists. For random input we have that $p=\Theta(n)$. Therefore, the running time of the algorithm solving $2$SUM reporting is $O(n^{1.5-0.5\delta}+op)$ and it uses $\tilde{O}(n^{\delta})$ space. Consequently, following our notation, we have that $f(\delta)=1.5-0.5\delta$.

Using this algorithm for $2$SUM we have the following corollaries of our previous theorems:

\begin{corollary}\label{cor:3sum_by_2sum}
There is a Monte Carlo randomized algorithm that solves $3$SUM on random input using $\tilde{O}(n^{\delta})$ space and $\tilde{O}(n^{2.5-1.5\delta}+n^2)$ time, under the assumption that there is a random read-only access to a random function $h : [m] \rightarrow [n]$.
\end{corollary}

\begin{corollary}
There is a Monte Carlo randomized algorithm that solves $4$SUM, $6$SUM and $8$SUM on random input using $\tilde{O}(n^{\delta})$ space and $\tilde{O}(n^{3.5-1.5\delta}+n^2)$, $\tilde{O}(n^{5.5-2.5\delta}+n^4)$ and $\tilde{O}(n^{7.5-3.5\delta}+n^{6-\delta})$ time, respectively, under the assumption that there is a random read-only access to a random function $h : [m] \rightarrow [n]$.
\end{corollary}

Notice that these last results are better than those achieved by Bansal et al.~\cite{BGNV17}.
One very interesting implication of Corollary~\ref{cor:3sum_by_2sum} is that we can solve $3$SUM on random input in $\tilde{O}(n^2)$ time using just $\tilde{O}(n^{1/3})$ space. This is better than the best currently known solution of $3$SUM that runs in $O(n^2)$ time using $O(n^{1/2})$ space.

\bibliographystyle{plainurl} \bibliography{ms}

\newpage
\appendix
\section*{Appendix}

\section{Balancing an Almost-Balanced Hash Function}\label{apx:balanced}
In our paper we use almost-balanced hash functions intensively. To make the presentation clearer we assumed that applying an almost-balanced hash function $h$ with range $[n^{m-\delta}]$ to $n^m$ possible $m$-tuples of some $m$ arrays with $n$ elements is expected to map $O(n^{\delta})$ elements to each value of $h$. However, as $h$ is not guaranteed to be a balanced hash function but rather an almost-balanced hash function, we can have $O(n^{m-\delta})$ elements that are mapped to overflowed values of $h$. In order for our Las Vegas algorithms throughout this paper to work we need to handle this overvalued values somehow even though we do not have enough space for saving all these elements. The idea to overcome these overflowed values is to eliminate them using a modified version of the technique by Wang~\cite{Wang14}. A useful observation by Wang~\cite{Wang14} implies, informally speaking, that if we apply an almost-balanced hash function $h$ to $n$ elements and its range is roughly $\sqrt{n}$, then there is a good probability that no value of $h$ has more than $c\sqrt{n}$ elements that are mapped to it, for some constant $c$ (see Corollary 2 in~\cite{Wang14}). Wang used this observation in order to map $n$ elements to sets of $O(n^{\delta})$ elements. This was done by creating a sequence of hash functions $(h_1,h_2,...,h_{\ell})$ such that each $h_i$ reduces the size of the set of elements to a square root of its size after applying the previous hash functions in the sequence. Specifically, if we begin with a set of $n$ elements the range of $h_1$ is $n^{1/2}$, the range of $h_2$ is $n^{1/4}$ and so on. The sequence ends when the number of elements in the each set is finally reduced to $O(n^{\delta})$ (see the full details in the proof of Theorem 2 in~\cite{Wang14}). Each of the hash functions in the sequence is checked to verify that it creates balanced sets (we cannot create the sets as we do not have enough space, instead we just check their size). If not, another hash function is picked and checked. Because of the observation, the number of hash functions we have to pick before finding a balanced one is expected to be three, if we choose $c$ to be approximately $k$. We emphasize that the next hash function in the sequence is chosen after fixing a specific value for the previous ones. Therefore, a different sequence is picked for every possible choice of values for the different hash functions in the sequence.

We use the same idea for balancing our almost-balanced hash functions. Whenever we begin with a set of $n^m$ elements (representing all $m$-tuples from $m$ arrays of $n$ elements), we reduce it to sets of size $O(n^{\delta})$ by using a sequence of hash functions $(h_1,h_2,...,h_{\ell})$ that each of them reduces the size of the sets to a square root of their size after applying the previous hash functions in the sequence. As in Wang's method, we check that each hash function creates balanced sets (using the same constant as the the one used by Wang and just counting the sizes of the implied sets) given a specific choice of values to the previous hash functions in the sequence and after these hash functions were already approved as being balanced.

We modify the method of Wang in two specific points:

(i) In Wang's solution, given an instance of $k$SUM with $k$ arrays he partition the arrays into $k$ groups that each of them contains a single array with $n$ elements. Therefore, he fixed the values of each hash function in the sequence for all groups together. In our general self-reduction, we have several groups of $m$ arrays. Following our "\emph{small space of fixed group}" observation (see Section~\ref{sec:self_reduction}), we just need to balance one group, while for all other groups the mapping of the $m$-tuples sums by the hash function can be unbalanced. Therefore, we have to fix the values of the hash functions for just \textbf{one} group and verify the sets are balanced within that group.

(ii) The second modification is due to the fact that we have a group of $m$ arrays and not just one array. Therefore, there can be duplicate sums, which we want to eliminate as explained in Section~\ref{sec:self_reduction}. Consequently, we do not want the number of duplicate sums to be included in our counting process. To handle this, we do not verify that each hash function implies balanced sets immediately after fixing the previous ones in the sequence and their values, but rather we choose all the hash functions in the sequence and then check if we get balanced sets of size $O(n^{\delta})$ for a specific choice of values to all hash functions. The crucial point is that the number of hash functions we have to pick for each position in the sequence is expected to be constant and the length of the sequence depends only on $k$ and $\delta$ which are also constants. Therefore, the number of sequences we have to check is expected to be $O(1)$. Checking each sequence can be done in $O(n^m)$ time by iterating over all $m$-tuples. The number of possible values to all hash functions in the sequence is $O(n^{m-\delta})$. Therefore, the total time for finding all the balanced hash functions we need is $O(n^{2m-\delta})$.

In the case of our self-reduction (see Section~\ref{sec:self_reduction}) the process of finding a sequence of hash function that balance one group is done before the recursion. Consequently, the running time is added outside the parenthesis of our recursion formula. Therefore, it does not affect the total running time of the self-reduction process following our choices of values for $m$. We also note that the same method can be applied to all other results in the paper that use almost-balanced hash functions without affecting the running time (following our analysis).

Using this method, we emphasize that in our Las Vegas algorithms each $h$ we pick is in fact a sequence of hash functions. This $h$ does not produce a single value but rather a vector of values (with constant length). In the recursion process of our self-reduction we regard the value of the hash function as this vector and apply the next hash function to each of its elements. In the base case of $2$SUM, we sort one array (or parts of this array) according to each position in the vector after sorting the previous ones (this is in fact a lexicographic order) and then use this sorting to quickly find a match to elements from the other array. To simplify this process, we can also map vectors to single values by picking $M$ which is the largest value any hash function in the sequence can get and multiply the first element in the vector by $M$, the second element by $M^2$ and so on, and then sum all these values. This simplified method can also be useful for a black box algorithm we may use like in Section~\ref{sec:3sum_random_input}. Finally, we also note that to get just a Monte Carlo solution we do not need this method of balancing, as the probability that an element within a solution to $k$SUM is in an overflowed value is small (this can be used, for example, in Section~\ref{sec:3sum_random_input}).

\end{document}